\documentclass[twocolumn,letterpaper,10pt]{article}
\usepackage{iasted}
\usepackage{times}
\usepackage[dvips]{graphicx}
\usepackage{url}

\makeatletter
\def\url@leostyle{%
  \@ifundefined{selectfont}{\def\UrlFont{\sf}}{%
    \def\UrlFont{\small\ttfamily}}\Url@do
}
\makeatother

\newcommand{\MapXMLelementToTableName}{
   \sigma
}
\newcommand{\MapXMLattributeToRelationalAttribute}{
   \theta
}
\newcommand{\MapLeafXMLElementToRelationalAttribute}{
   \delta
}

\newtheorem{DummyLemma}{DummyLemma}[section]
\newtheorem{theorem}[DummyLemma]{Theorem}
\newtheorem{lemma}[DummyLemma]{Lemma}
\newtheorem{ExamplE}[DummyLemma]{\em Example}
\newtheorem{DefinitioN}[DummyLemma]{\em Definition}

\newenvironment{definition}[1]
               {\begin{DefinitioN}[#1] \rm}{\end{DefinitioN}}

\newenvironment{proof}
               {\noindent {\bf Proof:} \rm}

\begin{document}

\title{Mapping XML Data to Relational Data: A DOM-Based Approach}

\author{Mustafa Atay, Yezhou Sun, Dapeng Liu, Shiyong Lu, Farshad Fotouhi\\
Department of Computer Science\\
Wayne State University, Detroit, MI 48202 \\
\{matay, sunny, dliu, shiyong, fotouhi\}@cs.wayne.edu\\
}

\date{}
\maketitle
\thispagestyle{empty}

\noindent {\bf\normalsize ABSTRACT}\newline {XML has emerged as the
standard for representing and exchanging data on the World Wide Web.
It is critical to have efficient mechanisms to store and query XML
data to exploit the full power of this new technology. Several
researchers have proposed to use relational databases to store and
query XML data. While several algorithms of schema mapping and query
mapping have been proposed, the problem of mapping XML data to
relational data, i.e., mapping an XML INSERT statement to a sequence
of SQL INSERT statements, has not been addressed thoroughly in the
literature. In this paper, we propose an efficient linear algorithm
for mapping XML data to relational data. This algorithm is based on
our previous proposed inlining algorithm for mapping DTDs to
relational schemas and can be easily adapted to other inlining
algorithms.}
 \vspace{2ex}

 \noindent
{\bf\normalsize KEY WORDS}\newline {XML, schema mapping, data
mapping, RDBMS.}
\vspace{-0.125in}
\section{Introduction}
XML is rapidly emerging as the de facto standard for representing and exchanging data over the World Wide Web. The increasing amount of XML documents requires the need to store and query XML documents efficiently. Researchers have proposed using relational databases to store and query XML documents ~\cite{e3}\cite{e5}\cite{e6}~\cite{e1}~\cite{e2}\cite{e14}. The main challenge of this approach is that, one needs to resolve the conflict between the hierarchical nature of XML data model and the two-level nature of relational data model. The following problems need to be addressed in order to employ relational databases to store and query XML data:
\begin{itemize}
\item
{\it Schema mapping}, which generates the corresponding relational schema from an input DTD. Instead of generating a relational table for each XML element, typically, several XML elements are combined into one table to reduce the number of generated tables and the cost of join operations. Representatives of these algorithms include the shared-inlining algorithm \cite{e11} and its variation \cite{e7}.
\item
{\it Data mapping}, which inserts XML data as relational tuples into the target database. Based on the relational schema generated in schema mapping, input XML documents are shredded and composed into relational tuples and inserted into the relational database. This requires that an XML INSERT statement be translated into a sequence of SQL INSERT statements, which are executed against the target database to load the data.
\item
{\it Query mapping}, which translates XML queries into SQL queries. Each XML query over XML documents needs to be translated into a sequence of SQL queries to be executed against the relational database.
\item
{\it Reverse data mapping}, which publishes XML data from relational data. XML queries are answered by executing the corresponding SQL queries which return relational data. These relational data need to be reformatted into XML data conforming to the structure imposed by the input XML query.
\end{itemize}

Numerous researchers have addressed the problems of schema mapping ~\cite{e7}~\cite{e11}~\cite{e14}, query mapping ~\cite{e2}~\cite{e11}~\cite{e12} and reverse data mapping ~\cite{e2}~\cite{e4}~\cite{e10}. However, the problem of data mapping is mostly ignored in the literature. In this paper, we address the data mapping problem. We propose an efficient linear algorithm to perform data mapping. This algorithm is based on our proposed schema mapping algorithm ~\cite{e7} (referred to by DTDMap afterwards) but can be easily adapted to other inlining algorithms such as the standard shared-inlining algorithm ~\cite{e11}

$Organization.$ The rest of the paper is organized as follows. Section 2 presents an overview of related work. Section 3 gives a brief overview of our schema mapping algorithm DTDMap. Section 4 identifies the data model and describes our proposed data mapping algorithm XInsert. Section 5 presents the experimental results of applying our data mapping algorithm to the two schema mapping algorithms, DTDMap and shared-inlining. Finally, Section 6 concludes the paper and points out some potential future work.
\vspace{-0.15in}
\section{Related Work}
Different approaches have been proposed for storing and querying XML
data. One approach is to develop native XML databases that support
XML data model and query languages directly. This includes Software
AG's Tamino XML Server ~\cite{e19}, IXIA's TEXTML Server ~\cite{e20}
and Sonic Software's eXtensible Information Server ~\cite{e18}
(formerly eXcelon's XIS). The advantage of this approach is that XML
data can be stored and retrieved in their original formats and no
additional mappings or translations are needed. Furthermore, most
native XML databases have the ability to perform sophisticated
full-text searches including full thesaurus support, word stubbing,
and proximity searches. The disadvantage is that, due to the
document-centric nature of these databases, complex searches or
aggregations might be cumbersome.

The second approach is to use the XML enabled commercial database
systems. Currently, most major databases such as SQL Server
~\cite{e17}, Oracle ~\cite{e8} and DB2  ~\cite{e16} provide
mechanisms to store and query XML data by extending the existing
data model with an additional XML data type so that a column of this
data type can be defined and used to store XML data. In addition, a
set of methods is associated with this new XML data type to process,
manipulate and query stored XML data.

The third approach is to use existing mature technologies such as
relational DBMSs or object-oriented DBMSs to store and query XML
data ~\cite{e3}~\cite{e7}~\cite{e11}. The main challenge of this
approach is that, one needs to resolve the conflict between XML data
model and the target data model. This usually requires various
mappings (e.g., schema mapping, data mapping and query mapping) to
be performed between the two data models. Therefore, the main issue
is to develop efficient algorithms to perform these mappings.

Different approaches have their pros and cons and the choice has to
be made based on the requirement of the application at hand and the
advancement of these approaches at the time that the choice has to
be made. Readers are referred to a recent evaluation study of
alternative XML storage strategies ~\cite{e13} for more details.

\vspace{-0.15in}
\section{An Overview of Schema Mapping Algorithm DTDMap}
The data mapping algorithm proposed in this paper is based on our schema mapping algorithm, DTDMap, proposed in ~\cite{e7}. In this section, we give a brief overview of DTDMap algorithm.

One approach for mapping DTDs to relational schema is mapping each node in the DTD to a table. Although this approach is easy to understand and implement, it has its own drawbacks. This approach results in many tables in the corresponding relational schema. When you query this database, you need to join several tables which causes query processing to be inefficient.

Alternatively, in our DTDMap algorithm, we suggest combining every single child node in a DTD, to its parent node, if it appears in its parent at most once. We call this operation $inlining$. A node is said to be $inlinable$ if it has exactly one parent node, and, its cardinality is not equal to either ``*'' or ``+''. An inlinable node is mapped with its parent node into the same table. Hence, we reduce the number of tables and consequently the average number of joins for queries.

DTDMap algorithm takes a DTD as input and produces a relational schema as output. In addition, it outputs mapping functions between XML elements and attributes in the input DTD, and, corresponding tables and relational attributes in the output schema.

\begin{figure}[ht]
\begin{center}
\includegraphics[scale=.75]{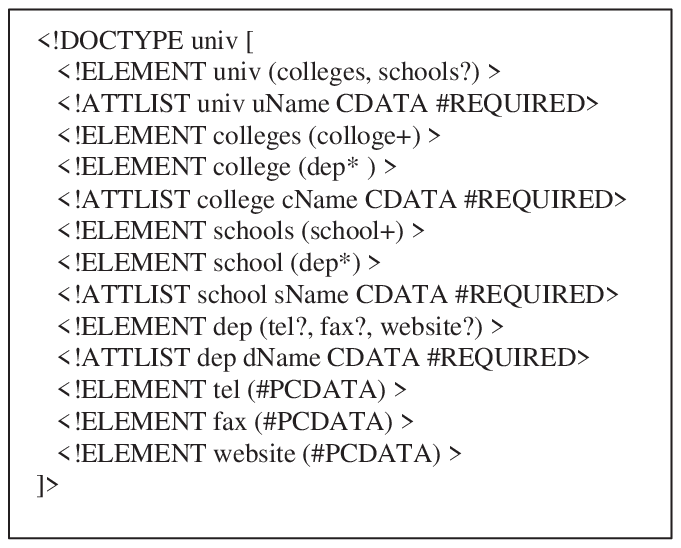}
\caption{Sample XML DTD "univ.dtd"}
\label{dtd}
\end{center}
\end{figure}

\vspace{-0.2in}

\begin{figure}[ht]
\begin{center}
\includegraphics[scale=.70]{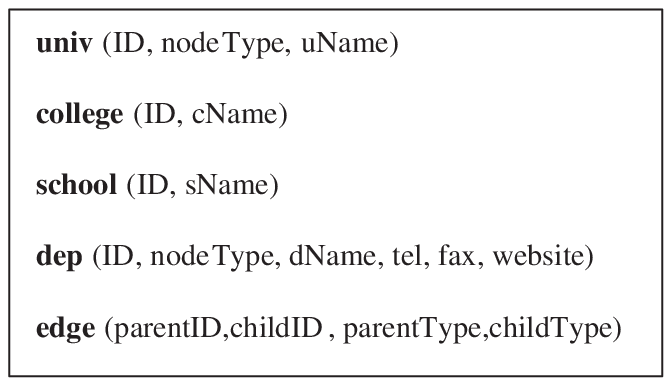}
\caption{Relational schema for univ.dtd}
\label{schema}
\end{center}
\end{figure}

\vspace{-0.2in}

\begin{figure}[ht]
\begin{center}
\includegraphics[scale=.75]{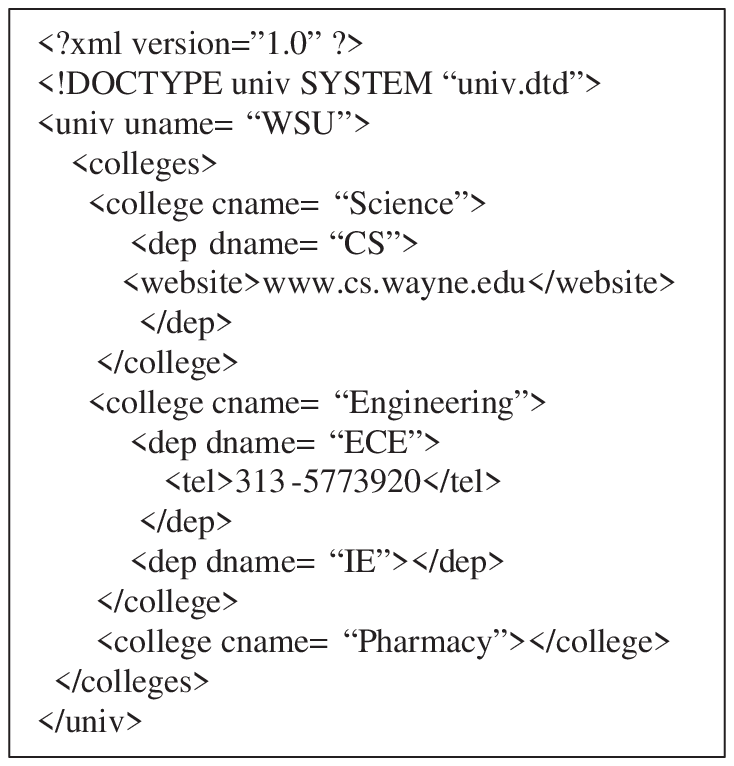}
\caption{Sample XML document univ.xml}
\label{instance}
\end{center}
\end{figure}

\vspace{-0.2in}

\begin{figure*}[t]
\begin{center}
\includegraphics[scale=.80]{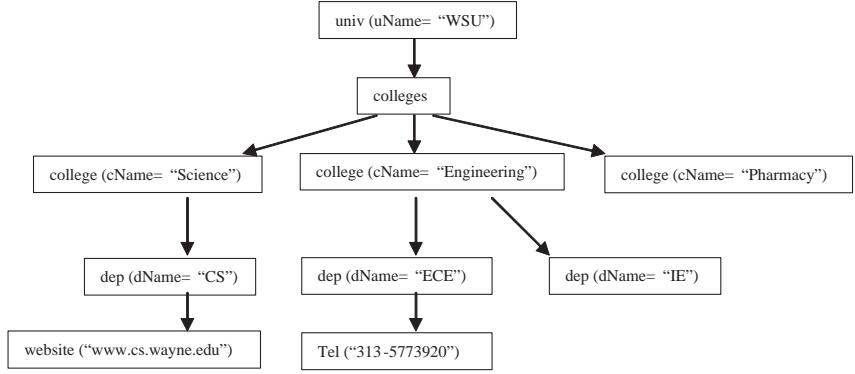}
\caption{DOMTree for univ.xml}
\label{domtree}
\end{center}
\end{figure*}

As an example, given the DTD in Figure 1, DTDMap generates the relational schema shown in Figure \ref{schema} and the following three mapping functions $\sigma$, $\theta$, and $\delta$:
\begin{itemize}
\item
$\sigma$(e) maps an element type to a corresponding relational table. Therefore, $\sigma$(univ) = univ, $\sigma$(uName) = univ, etc.
\item
$\theta$(a) maps an XML attribute to a relational attribute. Therefore, $\theta$(uName) = uName, $\theta$(dName) = dName, etc.
\item
$\delta$(e) maps a leaf element to a relational attribute. Therefore, $\delta$(tel) = tel, $\delta$(fax) = fax, etc.
\end{itemize}

In the above examples, some mappings happened to be identity
mappings. This is not always the case in practice; they can be
general enough and used to resolve name conflict. For example, a
mapping $\theta$(date) = BirthDate can avoid the use of ``date'' as
a column name for a table since ``date'' might be a keyword of the
target database and cannot be used as the name of a relational
attribute.
\vspace{-0.15in}
\section{Data Mapping}
The data model we will use for the data mapping algorithm is based
on the W3C's Document Object Model ~\cite{e21}, but it has some
distinctions from DOM specification. In contrast to traditional DOM
tree, the XML element DOM tree, which we propose here, does not
consider XML values as nodes but consider them as data fields of XML
element nodes. This distinction is only for the convenience of the
presentation; thus the algorithm proposed in this paper can be
implemented directly on the standard DOM model. The details of XML
element DOM tree are given in Definition 4.1.

\begin{definition}{DOMTree}
We model an XML element document $D$ as an XML element DOM tree
(DOMTree) $T$, in which nodes represent XML elements and edges
represent parent-child relationships between XML elements. For each
XML element node $e$ in $T$, we use the following notations:
\begin{itemize}
\item
$e.name$, the name of the XML element.
\item
$e.parent$, the parent node of $e$, and $e.parent = NULL$ if $e$ is the root node of $T$.
\item
$e.children$, the set of children nodes of $e$, and $e.children = \Phi$ if $e$ is a leaf of $T$.
We also denote the children of $e$ by $e.c_1, \cdots, e.c_m$.
\item
$e.attributes$, the set of XML attributes of $e$.
We also denote the attributes of $e$ by $e.a_1, \cdots, e.a_n$, and the names and values of
these attributes by $e.a_i.name$ and $e.a_i.value$ respectively ($i=1, \cdots, n$).
\item
$e.value$, the value of $e$, and $e.value = NULL$ if $e$ is a non-leaf node.
\end{itemize}
\end{definition}

\begin{figure*}[ht]
\centering
\scriptsize{
\framebox[4.6in][t]{
\begin{minipage}[t]{4.5in}
\begin{tabbing}
\hspace*{0.2in}\=\hspace*{0.2in}\=\hspace*{0.2in}\=\hspace*{0.2in}\=\hspace*{0.2in}\=\hspace*{0.2in}\=\\
1 {\bf Algorithm} XInsert\\
2 {\bf Input:} DOMTree $T$, DTDGraph $G$, Schema mappings $\MapXMLelementToTableName, \MapXMLattributeToRelationalAttribute, \MapLeafXMLElementToRelationalAttribute$ \\
3 {\bf Output:} elements in $T$ are inserted into the relational database\\
4 {\bf Begin}\\
5    \> Queue q := EmptyQueue(), T.root.EID.value := genID(), q.enqueue(T.root)\\
6   \> {\bf While} q.isNotEmpty() {\bf do}\\
7    \> \> e := q.dequeue()\\
8   \> \> Table tb:= $\MapXMLelementToTableName(e)$\\
9   \> \> create a tuple $tp$ of table $tb$ with all attributes initialized to NULL \\
10   \> \> tp.ID = e.EID.value \\
11    \> \> {\bf If} nodeType $\in$ tb.AttributeSet {\bf then}
            tp.nodetype = e.name
      {\bf End If}\\
12    \> \> {\bf For} each XML attribute $e.a_i$ of $e$ {\bf do}
            $tp.\MapXMLattributeToRelationalAttribute(e.a_i) := e.a_i.value$
         {\bf End For}\\
13    \> \> {\bf If} e is a leaf {\bf then}\\
14    \> \> \> $tp.\MapLeafXMLElementToRelationalAttribute(e) := e.value$\\
15    \> \> {\bf Else} /* $e$ is not a leaf */\\
16    \> \> \> Queue r := EmptyQueue()\\
17    \> \> \> {\bf For} each child $e.c_i$ of $e$ {\bf do}
                 $r.enqueue(e.c_i)$
             {\bf End For}\\
18    \> \> \> {\bf While} r.isNotEmpty() {\bf do}\\
19    \> \> \> \> f := r.dequeue()\\
20    \> \> \> \> {\bf If} $f$ is not inlinable to $e$ {\bf then} \\
21    \> \> \> \> \> $f.EID.value := genID()$\\
22    \> \> \> \> \> $f.parentEID.value := e.EID.value$\\
23    \> \> \> \> \> $f.parentNodeType.value := e.name$\\
24    \> \> \> \> \> {\bf If}  {\it l(type(f.parent), type(f), G)} $<$$>$ `*' {\bf then}\\
25    \> \> \> \> \> \>  $tp.\MapXMLattributeToRelationalAttribute(f.EID) := f.EID.value$\\
26    \> \> \> \> \> {\bf End If}\\
27    \> \> \> \> \> q.enqueue(f)\\
28    \> \> \> \> {\bf Else} /* $f$ inlinable to $e$ */ \\
29    \> \> \> \> \> {\bf for} each XML attribute $f.a_i$ of $f$ {\bf do}
                   $tp.\MapXMLattributeToRelationalAttribute(f.a_i) := f.a_i.value$
                {\bf End For}\\
30    \> \> \> \> \> {\bf If} f is a leaf {\bf then}\\
31    \> \> \> \> \> \> $tp.\MapLeafXMLElementToRelationalAttribute(f) := f.value$\\
32    \> \> \> \> \> {\bf Else} /* $f$ is not a leaf */\\
33    \> \> \> \> \> \>  {\bf For} each child $f.c_i$ of $f$ {\bf do} $r.enqueue(f.c_i)$ {\bf End For}\\
34    \> \> \> \> \> {\bf End If}\\
35    \> \> \> \> {\bf End If}\\
36    \> \> \> {\bf End While}\\
37    \> \> {\bf End If}\\
38    \> \> Insert tuple $tp$ into table $tb$\\
39    \> \> {\bf If}  {\it l(type(e.parent), type(e), G)} $==$ `*' {\bf then}\\
40    \> \> \> insert $<e.parentEID.value, e.EID.value, e.parentNodeType.value, e.name>$ \\
41    \> \> \> into table $edge$\\
42    \> \> {\bf End If}\\
43    \> {\bf End While}\\
44 {\bf End Algorithm}
\end{tabbing}
\end{minipage}
}
}
\caption{The algorithm for mapping XML data to relational data}
\label{algorithm}
\end{figure*}

An XML element DOM tree for the XML document given in Figure \ref{instance} is illustrated in Figure \ref{domtree}. Each node $e$ is labeled by $e.name(e.value,e.a_1.name= e.a_1.value,\cdots,e.a_n.name= e.a_n.value)$ and $e.value$ is omitted when $e$ is non-leaf node where $e.value$=NULL.

Our data mapping algorithm XInsert is based on the notion of inlinable elements introduced in the schema mapping phase. However we cannot find out whether a XML element instance is inlinable from a DOMTree itself. We need to refer to the DTD information.

Figure \ref{algorithm} describes the algorithm XInsert which inserts an XML document into the relational database whose schema is previously generated from the input XML DTD.

Given a DTD graph $G$ and nodes $n_1$ and $n_2$ in $G$, $l(n1,n2,G)$ denotes the label of edge between $n_1$ and $n_2$. Given an element instance $e$, $type(e)$ denotes the corresponding element node in $G$.

We define a field EID  which is associated with each element instance $e$ in the algorithm. EID is a unique value and it is generated for each non-inlinable element when it is first visited. We introduce $parentEID$ and $parentNodeType$ fields in the algorithm to keep the parent-child relationship between the elements.

XInsert algorithm is driven by two nested While loops. The outer
While loop maintains a Queue, $q$, to process the non-inlinable XML
elements. It obtains the typical information of the tuple, $t$,
corresponding to a non-inlinable element, $e$, such as ID, nodeType,
XML attribute values and the content(lines 10-14). Finally, it
inserts the tuple $t$ into the table $\sigma(e)$ (line 38). If
$type(e)$ is a *-element, then it inserts the tuple, $t_e$ into the
$edge$ table to store the parent-child relationship (lines 39-42).

If $e$ is not a leaf element then the inner While loop is performed to search for inlinable descendants of $e$. It maintains a Queue, $r$, to  process the descendants of $e$. Firstly, it determines the inlinable descendants of $e$ and retrieve their data to complete the context information for the tuple $t$ (lines 28-35). Secondly, it keeps the parent-child information of the non-inlinable descendants of $e$ through the fields $parentEID$ and $parentNodeType$ (lines 20-23). Lastly, it introduces a foreign key in the tuple $t$ if there exists a shared descendant of $e$ (lines 24-26).  The algorithm ends when there are no more elements to be processed in queues, $q$ and $r$.

\begin{table*}[htbp]
\scriptsize{
\caption{The state of the database after univ.xml is stored}
\makebox[1cm][t]{}\hfill
\begin{minipage}[t]{3.5cm}
\begin{tabular}[t]{|l|l|l|}\hline
\multicolumn{3}{|l|}{\bf Univ}\\ \hline
\underline{{\it ID}} & {\it nodeType} & {\it uName} \\ \hline
1 & univ & WSU \\ \hline
\end{tabular}
\end{minipage}
\hfill
\begin{minipage}[t]{5.5cm}
\begin{tabular}[t]{|l|l|l|l|}\hline
\multicolumn{4}{|l|}{\bf Edge}\\ \hline
\underline{{\it parentID}} & \underline{{\it childID}} & {\it parentType} & {\it childType} \\ \hline
1 & 2 & univ & college \\ \hline
1 & 3 & univ & college \\ \hline
1 & 4 & univ & college \\ \hline
2 & 5 & college & dep \\ \hline
3 & 6 & college & dep \\ \hline
3 & 7 & college & dep \\ \hline
\end{tabular}
\end{minipage}
\hfill
\begin{minipage}[t]{4cm}
\begin{tabular}[t]{|l|l|}\hline
\multicolumn{2}{|l|}{\bf School}\\ \hline
\underline{{\it ID}} & {\it sName} \\ \hline
\end{tabular}
\end{minipage}
\\
\\
\mbox{}\hfill
\begin{minipage}[t]{8.5cm}
\begin{tabular}[t]{|l|l|l|l|l|l|}\hline
\multicolumn{6}{|l|}{\bf Dep}\\ \hline
\underline{{\it ID}} & {\it nodeType} & {\it dName} & {\it tel} & {\it fax} & {\it website} \\ \hline
5 & dep & CS & null & null & www.cs.wayne.edu \\ \hline
6 & dep & ECE & 313-5773920 & null & null \\ \hline
7 & dep & IE & null & null & null \\ \hline
\end{tabular}
\end{minipage}
\hfill
\begin{minipage}[t]{4.5cm}
\begin{tabular}[t]{|l|l|}\hline
\multicolumn{2}{|l|}{\bf College}\\ \hline \underline{{\it ID}} &
{\it sName} \\ \hline 2 & Science\\ \hline 3 & Engineering\\ \hline
4 & Pharmacy\\ \hline
\end{tabular}
\end{minipage}
\label{database}
}
\end{table*}

To analyze the time complexity of algorithm XInsert, we first present some properties of the algorithm in the following lemmas.

\begin{lemma}
Each non-inlinable element $e$ in DOMTree $T$ is enqueued into Queue $q$ exactly once, and $q$ only contains non-inlinable elements.
\end{lemma}
\begin{proof}
The operation of enqueue into $q$ is performed only at line 5 and at line 27. Line 5 enqueues the root element which is non-inlinable. Line 27 is in the body of if-statement (line  20) whose condition indicates the element $f$ to be enqueued into $q$ in line 27 is non-inlinable. Therefore, $q$ only contains non-inlinable elements.
\end{proof}

To demonstrate that each non-inlinable element $e$ is enqueued into $q$ exactly once, we prove in the following that $e$ is enqueued into $q$ at most once and at least once respectively. First, we notice that for each element $e$ that is dequeued from $q$, $e$ is non-inlinable as $q$ only contains non-inlinable elements, the While-statement (line 18 to 36) will enqueue each of $e$'s descendant element $f$ exactly once into Queue $r$, where $f$ satisfies (1) f is $e$'s child (line 17) , or (2) $f$ is inlinable to $e$ (line 33), or (3) $f$ is non-inlinable to $e$ but $f$'s parent is inlinable to $e$ (line 33).  The acyclicity of $T$ implies that each non-inlinable element of $T$ can be enqueued into $q$ at most once. In addition, except the root element, the While-statement (line 18 to 36) will ensure that each non-inlinable element will be enqueued into $q$ at least once in line 27. Finally, the root element is enqueued into $q$ exactly once. Therefore, each non-inlinable element $e$ is enqueued into $q$ exactly once.

\begin{lemma}
Each XML element $e$, except the root element in DOMTree $T$ is enqueued into Queue $r$ exactly once.
\end{lemma}
\begin{proof}
Lemma 4.2 implies that each non-inlinable element $e$ is dequeued from $q$ exactly once (line 7), and for each such $e$, the While-statement (line 18 to 36) will enqueue each of $e$'s descendant element $f$ exactly once into Queue $r$, where $f$ satisfies (1) $f$ is $e$'s child (line 17), or (2) $f$ is inlinable to $e$ (line 33), or (3) $f$ is non-inlinable to $e$ but $f$'s parent is inlinable to $e$ (line 33).  Therefore, each element of $T$ except the root element will satisfy one of these three cases for some $e$ and thus will be enqueued into $r$ at least once.  The acyclicity of $T$ implies that each element of $T$ can be enqueued into $r$ at most once. Therefore, each XML element in $T$ is enqueued into $r$ exactly once.
\end{proof}

The following theorem demonstrates that XInsert is an efficient linear algorithm.

\begin{theorem}
The time complexity of algorithm XInsert is O(n) where $n$ is the number of XML elements and attributes in DOMTree $T$.
\end{theorem}

\begin{proof}
From Lemma 4.2, it follows that the While loop statement in line 6 will be executed for $m1$ times, where $m1$ is the number of non-inlinable elements in $T$. From Lemma 4.3, it follows that the While loop statement in line 18 will be executed for $m2$ times where $m2=n-1$. We have $m1 <= m2 < n$. All the operations involved in those two While loops spend constant amount of time. Hence, it is clear that the XInsert algorithm runs in $O(n)$ time complexity.
\end{proof}

\begin{figure}[ht]
\begin{center}
\includegraphics[scale=.70]{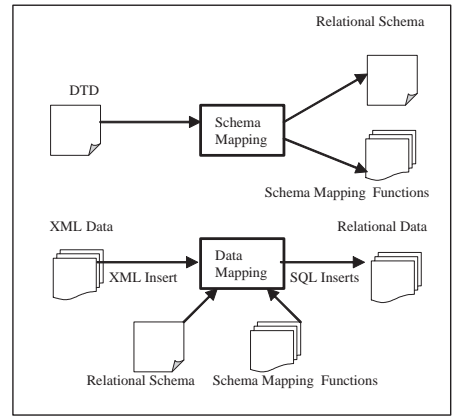}
\caption{System Architecture}
\label{architecture}
\end{center}
\end{figure}

Table \ref{database}  shows how the XML document given in Figure
\ref{instance} is mapped into the relational database given in
Figure \ref{schema} by the data mapping algorithm explained above.

We define the overall mapping process as comprised of schema mapping and data mapping modules. The system architecture is given in Figure \ref{architecture}.
\vspace{-0.15in}
\section{Experimental Results}
We applied the data mapping algorithm introduced in this paper to two different XML schema mapping schemes. The first XML schema mapping scheme that we use in our experiment is DTDMap algorithm and the second scheme is the shared inlining algorithm. These schemes map the XML elements to relational tables based on the operation of inlining child nodes into parent nodes.

We chose a DTD with document-centric features and another one with data-centric features for our experiment to show the performance of our algorithm on documents with different features. We took the first DTD, which is $catalog.dtd$, from XBench ~\cite{e15} XML Benchmark project and generate our test documents. This DTD shows data-centric features. The second DTD, which is $auction.dtd$, was taken from XMark ~\cite{e9} XML Benchmark project and the test documents were generated. The second DTD has document-centric features besides its data-centric features.

We used a Pentium IV computer with 2.4 GHz processor and 512 MB main memory. We ran our implementation using Java 1.4.1 software development kit. We maintained DOM element tree using W3C's DOM specification and processed it using an on-shelf DOM API.

We generated test documents in five different sizes from 10 MB to 50 MB. Then we applied test runs for each size of test documents for both DTDs using both mapping schemes. We created flat comma-separated text files for each table of the corresponding relational schema. Our performance metric is the time spent to map XML data to relational data. Loading data to the database is not included in this time. In order to see the pure performance of our data mapping algorithm, we did not populate a database directly.

We minimized the usage of system resources during the experiments to get more realistic spent time values. On the other hand we repeated every experiment for five times and got the mean value of spent time to obtain more accurate results.Table \ref{results} shows the time spent for data mapping in seconds.

\vspace{-0.2in}

\begin{table}[ht]
\caption{The time spent for XInsert data mapping algorithm}
\begin{center}
\includegraphics[scale=0.60]{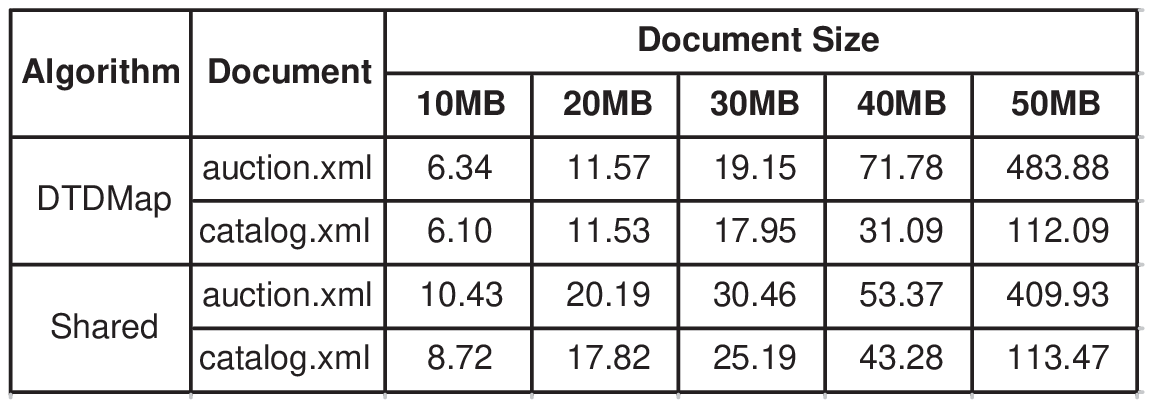}
\label{results}
\end{center}
\end{table}

\vspace{-0.2in}

We observe from the above table that both mapping schemes for both DTDs make a pick after 40 MB. Document Object Model loads the whole document tree into the main memory and then make the whole tree traversal available. When the document gets larger, it hardly fits into the memory. Then system stores some part of the tree in the disk and starts to come back and forth between the disk and the main memory which causes the increase in time.

We see that processing $auction.xml$ document takes more time than processing $catalog.xml$ document on average. We see the difference more precisely especially at 50 MB for both mapping schemes. The $auction.xml$ document includes document-centric textual elements which are nested recursively and cause DOMTree to be deeper where there is no recursive element in $catalog.xml$.
\vspace{-0.15in}
\section{Conclusions}
Several algorithms have been proposed for schema mapping by
researchers, but the problem of data mapping has not been discussed
thoroughly in the literature. In our study, we have addressed the
problem of data mapping and defined an efficient and linear
algorithm for mapping XML data to relational data. Our algorithm
populates a relational database with the input XML documents,
according to the relational schema generated by the schema mapping
phase. This algorithm can be easily adapted to other mapping schemes
based on inlining technique.

Ordered nature of XML elements are ignored in this study. We have not dealt with the referential and integrity constraints in our data mapping algorithm. These issues need to be investigated as a potential future work.
\bibliographystyle{latex8}
\vspace{-0.15in}
\bibliography{atay04}
\end{document}